\documentclass{llncs}
\renewcommand{\vec}[1]{\mathbf{#1}}
\newcommand{\field}[1]{\mathbb{#1}}

\usepackage{amsfonts}
\usepackage{amsmath}
\usepackage{amssymb}  
\usepackage{mathrsfs} 

\title{Determining Edge Expansion and Other Connectivity Measures of Graphs of Bounded Genus}
\author{Viresh Patel \thanks{Supported by EPSRC grant EP/F064551/1}}
  
\institute{School of Engineering and  Computing Sciences, Durham University, \\
Science Laboratories, South Road, Durham DH1 3LE, U.K. \\
\email{ viresh.patel@dur.ac.uk } }
\date{}

\begin{document}
\maketitle

\begin{abstract}
In this paper, we show that for an $n$-vertex graph $G$ of genus $g$, the edge expansion of $G$ can be determined in time $n^{O(g^2)}$. We show that the same is true for various other similar measures of edge connectivity. 
\end{abstract}

\section{Introduction}

\subsection{Background and Motivation}
Edge expansion (known also as the minimum cut quotient, the isoperimetric number, or the flux of a graph) is a well-studied notion in graph theory and arises in several contexts of discrete mathematics and theoretical computer science. These include the explicit construction of expander graphs, the analysis of certain randomised algorithms, and graph partitioning problems. In this paper, we are concerned with giving an exact algorithm for determining the edge expansion (and other similar measures) of graphs embedded on surfaces.

Throughout, we use the term graph to mean multigraph without loops, unless otherwise stated. For a graph $G=(V,E)$ and $e \in E$, we write $e=ab$ to mean that the vertices $a$ and $b$ are the end points of $e$. 

For $S$ a nonempty proper subset of $V$ and $\bar{S}$ its complement, we define
\[ [S, \bar{S}]_G = \{ e \in E: \: e=ab, \: a \in S, \: b \in \bar{S} \},
\]
which we  call an \emph{edge-cut} of $G$. (The subscript is dropped when it clear which graph we are referring to.)
For a cut $[S,\bar{S}]$ of a graph $G=(V,E)$, define the \emph{balance} of the cut to be $b(S,\bar{S}):= \min(|S|, |\bar{S}|)/|V|$. Note that the balance of a cut is a real number in the interval $(0,\frac{1}{2}]$. Two well-known graph cut problems, which take into account the balance of cuts, are the \emph{minimum quotient cut} problem and the \emph{sparsest cut} problem. These ask respectively to minimize the \emph{cut quotient} $q(S,\bar{S})$ and the \emph{cut density} $d(S,\bar{S})$ over all cuts $[S,\bar{S}]$ of a graph $G$, where
\[ q(S,\bar{S})=\frac{|[S,\bar{S}]|}{b(S,\bar{S})} \:\:\:\:\text{ and }\:\:\:\: 
d(S, \bar{S}) = \frac{|[S,\bar{S}]|}{b(S,\bar{S})(1 - b(S,\bar{S}))}.
\] 
Both $q$ and $d$ penalise unbalanced cuts, although $q$ does so to a greater extent than $d$. 

Problems such as the minimum quotient cut problem and the sparsest cut problem underlie many divide and conquer algorithms \cite{shm}, and find applications in VLSI layout problems, packet routing in distributed networking, clustering, and so on. Unfortunately, for general graphs, finding a minimum quotient cut or a sparsest cut is known to be NP-hard \cite{Gar,Mat}. Thus there are two possible ways of developing efficient algorithms for these problems: either by considering approximation algorithms or by restricting attention to certain graph classes. There has been much research done in finding approximation algorithms for these problems. Here, we mention only the seminal paper of Leighton and Rao \cite{Rao} giving a polynomial-time $O(\log n)$-approximation algorithm for the minimum quotient cut problem, and the significant improvement in the approximation factor to $O(\sqrt{\log n})$ in a paper of Arora, Rao, and Vazirani \cite{ARV}. On the hardness side, Amb{\"u}hl et al.\ \cite{AMS} proved that the sparsest cut problem admits no polynomial-time approximation scheme unless NP-hard problems can be solved in randomized subexponential time. 

We approach the problems of minimum quotient cut and sparsest cut from the perspective of developing exact polynomial-time  algorithms for restricted graph classes. Such approaches have not received as much attention in recent years as the development of approximation algorithms, but we hope this paper will take a step towards sparking interest.

Bonsma \cite{Bon} gave polynomial-time algorithms for finding sparsest cuts of unit circular graphs and cactus graphs. Park and Phillips \cite{Park}, building on the work of Rao \cite{Rao2}, gave a polynomial-time algorithm for determining the minimum quotient cut (as it has been defined above) of planar graphs. Given that many planar-graph algorithms have been adapted for generalizations of planar graphs -- see for example the introduction to \cite{CEN} and the references therein -- surprisingly little is known about the complexity of computing minimum quotient cuts or sparsest cuts for generalizations of planar graphs. Here, we generalize the algorithm of Park and Phillips to give the first exact polynomial-time algorithm for determining minimum quotient cuts and sparsest cuts of bounded-genus graphs.

\subsection{Results}

Before we state our result precisely, we give a generalization of the minimum quotient cut and sparsest cut. Notice that the denominators for both the cut quotient $q$ and the cut density $d$ are concave and increasing functions of $b(S,\bar{S})$ on the interval $[0,\frac{1}{2}]$.  For any concave, increasing function $f: [0,\frac{1}{2}] \rightarrow  [0, \infty)$ and a cut $[S, \bar{S}]$ of a graph $G$, we define
\[ d^f_G(S,\bar{S}) = \frac{|[S,\bar{S}]|}{f\big( b(S,\bar{S}) \big)},
\]
and we define
\[ d^f(G) = \min d^f_G(S, \bar{S}), \]
where the minimum is taken over all cuts $[S, \bar{S}]$ of $G$. Any cut $[S, \bar{S}]$ that minimizes $d^f_G$ is referred to as an \emph{$f$-sparsest cut} of $G$.

Let $f: [0, \frac{1}{2}] \rightarrow [0, \infty)$ be a fixed concave, increasing function that is computable in polynomial time on the rationals, and let $g$ be a fixed non-negative integer. The input for our algorithm is an $n$-vertex undirected multigraph $G$ of genus $g$. Our algorithm computes an $f$-sparsest cut of $G$ in time $O(n^{2g^2+4g+7})$.

\subsection{Overview and Techniques}

In this section we give an informal overview of our methods. Our methods extend those of Park and Phillips \cite{Park} and combine them with surface homology techniques, used for example in \cite{CEN}. 

A simple averaging argument shows that, given a graph $G$, there exists a sparsest cut $[S, \bar{S}]$ of $G$ that is \emph{minimal}, i.e.\ a cut where the graphs induced by $G$ on $S$ and $\bar{S}$ are both connected. This extends easily to $f$-sparsest cuts, where $f$ is a concave increasing function.

There is a standard correspondence between the cuts of a planar graph $G$ and the cycles of its dual $D(G)$: the minimal cuts of $G$ correspond precisely to the cycles of $D(G)$, and the size of a cut in $G$ is equal to the length of its corresponding cycle in $D(G)$. One can similarly construct a dual graph $D(G)$ for a graph $G$ embedded on a surface; however the correspondence between cuts of $G$ and cycles of $D(G)$ is not quite so simple. Roughly, for a graph $G$ of genus $g$, a cut of $G$ corresponds to a union of at most $g+1$ cycles of $D(G)$, but the reverse does not hold: a union of at most $g+1$ cycles in $D(G)$ does not necessarily correspond to a cut in $G$. Using the surface embedding of $G$, we construct a function $\Theta$ from the set of oriented edges of $D(G)$ to $\field{Z}^{2g}$ with the following property: summing $\Theta$ around the oriented edges of a union of cycles of $D(G)$ gives the zero vector if and only if that union of cycles corresponds to a (certain generalization of a) cut of $G$.

Extending and simplifying an idea from \cite{Park}, we also construct a function $\hat{w}$ from the set of oriented edges of $D(G)$ to $\field{Z}$ with the following property: if a union of cycles in $D(G)$ corresponds to a cut in $G$, then summing $\hat{w}$ around the oriented edges of cycles in the union essentially gives the balance of $[S,\bar{S}]$.

Using $D(G)$, $\Theta$, and $\hat{w}$, we construct a type of covering graph $H$, again extending an idea in \cite{Park}. For each fixed value $\vec{v}$ and $k$ of $\Theta$ and $\hat{w}$, we can use $H$ to find a shortest cycle in $D(G)$ whose $\Theta$-value is $\vec{v}$ and whose $\hat{w}$-value is $k$. Such a shortest cycle of $D(G)$ corresponds to a shortest path in $H$ between suitable vertices.

By repeatedly applying a shortest-path algorithm to $H$, we obtain, for every $\vec{v}$ and $k$ (in a suitable range), a shortest cycle of $D(G)$ whose $\Theta$-value is $\vec{v}$ and whose $\hat{w}$-value is $k$. We construct the set $X$ of every union of at most $g+1$ of these shortest cycles. The size of $X$ is $n^{O(g^2)}$, and we show that at least one element of $X$ corresponds to an $f$-sparsest cut of $G$.

\section{Preliminaries}

We begin this section by proving some simple inequalities for concave functions. Throughout, rather than working with concave increasing functions $f:[0, \frac{1}{2}] \rightarrow [0, \infty)$, we work instead with functions $f:[0,1] \rightarrow  [0, \infty)$ that are concave and increasing on $[0, \frac{1}{2}]$ and have the property that $f(x)=f(1-x)$ for all $x \in [0,1]$. We work with these functions purely for the convenience of having, for any cut $[S, \bar{S}]$ of $G$, that
\[ f(b(S,\bar{S})) = f(|S|/|V|) = f(|\bar{S}|/|V|).
\]
Note that such functions are in fact concave on their entire domain. For the algorithm, we assume that $f$ can be computed in polynomial time on the rationals.

\begin{lemma} \label{Inequality} Let $f:[0,1] \rightarrow  [0, \infty)$ be a concave increasing function on $[0, \frac{1}{2}]$ with the property that $f(x)=f(1-x)$ for all $x \in [0,1]$.
Suppose $x_1, \ldots, x_k \in [-1,1]$ and $x:= \sum_{i=1}^k x_i \in [-1,1]$. Then
\[ f(|x|) \leq  \sum_{i=1}^k f(|x_i|).
\]
\end{lemma}
\begin{proof}
We may assume without loss of generality that $x \in [0,1]$ by switching the signs of the $x_i$ if necessary. We also know that $f$ is concave on its entire domain. Recall that a function $f:[0,1] \rightarrow [0, \infty)$ is concave if and only if, for every $a,b,t \in [0,1]$, we have that $tf(a) + (1-t)f(b) \leq f(ta + (1-t)b)$. 

We prove the inequality using induction. The case $k=1$ is trivial. We prove the case $k=2$. We have, without loss of generality, the two cases 
\begin{itemize}
\item[(i)] $x_1, x_2 \in [0,1]$ and  
\item[(ii)] $x_1 \in [0,1]$, $x_2 \in [-1,0]$.
\end{itemize}

\emph{Case (i):}
Given $x_1, x_2 \in [0,1]$ with $x = x_1+x_2 \in [0,1]$, choose $r,s \in [0,1]$ such that $x_1 = rx$ and $x_2 = sx$; thus $r+s=1$. From the concavity of $f$, we have
\[ rf(x) + (1-r)f(0) \leq f(x_1)  \:\:\:\: \text{ and } \:\:\:\:  sf(x) + (1-s)f(0) \leq f(x_2).
\]
Adding the two inequalities together, and using the fact that $r+s=1$, we obtain
\[ f(x_1) + f(x_2) \geq f(x) + f(0) \geq f(x),
\] 
as required.

\emph{Case (ii):} 
Given $x_1 \in [0,1]$, $x_2 \in [-1,0]$ with $x = x_1+x_2 \in [0,1]$ apply case (i) to the numbers $1-x_1, -x_2 \in [0,1]$. Thus, we have
\[ f(|x_1|) + f(|x_2|) = f(1-x_1) + f(-x_2) \geq f(1-x_1-x_2) = f(1-x) = f(x),
\]
as required. This proves the case $k=2$. 

The induction step follows easily. Indeed, order the $x_i$ such that $\sum_{i=1}^rx_i \in [-1,1]$ for all $r = 1, \ldots, k$. Then
\begin{align*}
f(x) = f\Big( \sum_{i=1}^k x_i \Big) &=  f\Big( \sum_{i=1}^{k-1} x_i + x_k \Big) \\
& \leq f\Big( \Big| \sum_{i=1}^{k-1} x_i \Big| \Big) + f(|x_k|)  \tag{by case $k=2$} \\
&  \leq \sum_{i=1}^{k} f(|x_i|). \tag{by induction hypothesis}
\end{align*}
\hfill$\Box$
\end{proof}

Next we prove that for any graph $G$, $d^f_G$ can be minimized by a \emph{minimal cut}, that is, a cut $[S, \bar{S}]$ for which both $G[S]$ and $G[\bar{S}]$ are connected (here $G[A]$ denotes the graph induced by $G$ on $A \subseteq V$). This is a well-known fact for the cut quotient $q$ and the cut density $d$, and generalises easily to $d^f_G$. Before we can do this, we need a trivial averaging argument; we state it formally so that we can refer to it later.

\begin{proposition} \label{Ave}
For $i=1, \ldots, k$, let $a_i$, $p_i$, and $q_i$ be non-negative real numbers with at least one $q_i$ non-zero. Then there exists some $i' \in \{1, \ldots, k \}$ such that $q_{i'} \not= 0$ and
\[ \frac{p_{i'}}{q_{i'}} \leq \frac{ \sum_{i=1}^k a_ip_i }{ \sum_{j=1}^k a_jq_j }.
\]
\end{proposition} 
\begin{proof}
Choose $i'$ to minimize $p_{i'}/q_{i'}$ (if $q_i=0$, we take $p_i/q_i$ to be $\infty$). Thus we have that $p_{i'}/q_{i'} \leq p_{i}/q_{i}$ for all $i= 1, \ldots, k$. Rearranging, we have that $p_{i'}(a_iq_i) \leq (a_ip_i)q_{i'}$ for all $i= 1, \ldots, k$. Summing both sides over $i$ and rearranging gives the desired inequality.
\hfill$\Box$
\end{proof}

\begin{proposition} \label{Concut}
Let $G=(V,E)$ be a connected graph and let $f:[0,1] \rightarrow [0, \infty)$ be as in Lemma~\ref{Inequality}. Then there exists a cut $[S, \bar{S}]$ of $G$ such that $d^f_G(S, \bar{S}) = d^f(G)$ and for which $G[S]$ and $G[\bar{S}]$ are both connected. 
\end{proposition} 
\begin{proof}
Amongst all cuts of $G$ minimizing $d^f_G$, let $[A,B]$ be one of minimum size. Suppose that $G[A]$ has components $A_1, \ldots, A_r$ and $G[B]$ has components $B_1, \ldots, B_s$. We claim that there is some $i'$ for which $d^f_G(A_{i'},\bar{A_{i'}}) \leq d^f_G(A,B)$ (and by symmetry, there is some $j'$ for which $d^f_G(B_{j'},\bar{B_{j'}}) \leq d^f_G(A,B)$). Since $G$ is connected and $[A,B]$ is of minimum size, the claim implies that both $G[A]$ and $G[B]$ have only a single component, proving the proposition.
It remains to prove the claim.
Let $|V|=n$, $a_i = |A_i|/n$, and $b_i = |B_i|/n$. Let $a=|A|/n = \sum_{i=1}^ra_i$ and $b=|B|/n = \sum_{i=1}^sb_i$. Let
\[ e_{ij}= |\{ ab \in E: a \in A_i, b \in B_j \}| \:\:\: \text{ and } \:\:\: 
d^f_{ij}= \frac{e_{ij}}{f(a_i)}.
\] 
We have that
\begin{align*}
d^f(G) = d^f_G(A,B) = \frac{|[A,B]|}{f\big( b(A,B) \big)} &= \frac{\sum_{i=1}^r\sum_{j=1}^se_{ij}}
{f\big(\sum_{i=1}^ra_i\big)} \\
&= \frac{\sum_{i=1}^r f(a_i) \sum_{j=1}^s d^f_{ij} }
{f\big(\sum_{i=1}^ra_i\big)} \\
&\geq  \frac{\sum_{i=1}^r f(a_i) \sum_{j=1}^s d^f_{ij} }
{\sum_{i=1}^rf(a_i)},
\end{align*}
where the inequality follows from Lemma~\ref{Inequality}.
Applying Proposition~\ref{Ave}, we find that there exists some $i'$ such that
\begin{align*}
d^f(G) = d^f_G(A,B) 
 \geq  \frac{f(a_{i'}) \sum_{j=1}^s d^f_{i'j} }{f(a_{i'})} 
= \frac{ \sum_{j=1}^s e_{i'j}}{f(a_{i'})} 
= \frac{ [A_{i'}, \bar{A_{i'}}] }{f(a_{i'})}
= d^f_G(A_{i'}, \bar{A_{i'}}).
\end{align*}
\hfill$\Box$
\end{proof}

It turns out that the description and the proof of correctness of our algorithm is most conveniently and naturally expressed in the language of surface homology. Through the remainder of this section,  we introduce the necessary concepts keeping our treatment as simple and self-contained as possible. One can find more comprehensive treatments in e.g.\ \cite{Gib,GT}.

Although we are only concerned with undirected graphs when determining quotient cuts, sparsest cuts, and other vulnerability measures, we shall have cause to orient edges of our graph through the course of our proofs and algorithms. Each edge $e = ab$ of a graph $G=(V,E)$ has two orientations, namely $(a,e,b)$ and $(b,e,a)$. The two orientations are denoted $\overrightarrow{e}$ and $\overleftarrow{e}$, although we cannot say which is which in general. Given a set $E$ of edges, we write $\overrightarrow{E}$ for the set of their orientations, two for each edge. The \emph{edge space} of $G$, denoted $\mathcal{E}(G)$, is the free abelian group on $\overrightarrow{E}$ modulo the relation that $\overleftarrow{e} = -\overrightarrow{e}$. For each $\rho  \in \mathcal{E}(G)$, we can express $\rho$ uniquely as 
\[\rho = \sum_{\overrightarrow{e} \in \overrightarrow{E}} \lambda_{\overrightarrow{e}}\overrightarrow{e},\]
 where $\lambda_{\overrightarrow{e}} \in \field{Z}$ for all $\overrightarrow{e} \in \overrightarrow{E}$ and $\min(\lambda_{\overrightarrow{e}}, \lambda_{\overleftarrow{e}}) = 0$. Then, we define
\[ |\rho| = \sum_{\overrightarrow{e} \in \overrightarrow{E}} \lambda_{\overrightarrow{e}}.
\]
Thus $|\rho|$ in a sense counts the number of edges in $\rho$.

The \emph{cut space} $\mathcal{T}(G)$ of $G=(V,E)$, which is a subgroup of $\mathcal{E}(G)$, is defined as follows. For a cut $[S,\bar{S}]$ of $G$, define 
\[ \overrightarrow{[S,\bar{S}]} = \sum_{ab = e \in E  \atop  a \in S, \: b \notin S} (a,e,b).
\]
Then $\mathcal{T}(G)$ is the subgroup of $\mathcal{E}(G)$ generated by $\{\overrightarrow{[S,\bar{S}]}: S \subseteq V \}$. 

The next lemma shows how, by suitably assigning weights to oriented edges of a graph, we can determine the balance of a cut simply by summing the weights of the edges in the (oriented) cut. This is a generalisation of a result for planar graphs that was presented (but not proved) in \cite{Park}.

\begin{lemma} \label{Weight}
Let $G=(V,E)$ be a connected graph and let $v \in V$ be some fixed vertex. There exists a function $w:\overrightarrow{E} \rightarrow \field{Z}$ with the following properties. 
\begin{itemize}
\item[(i)]For every $\overrightarrow{e} \in \overrightarrow{E}$, we have $w(\overleftarrow{e})=-w(\overrightarrow{e})$. 
\item[(ii)] For every $\overrightarrow{e} \in \overrightarrow{E}$, we have $|w(\overrightarrow{e})| \leq |V|$.
\item[(iii)] For every $S$ satisfying $v \in S \subseteq V$, we have that
\[ w( \overrightarrow{[S,\bar{S}]} ):= \sum_{(a,e,b) \in \overrightarrow{[S, \bar{S}]} }w(a,e,b) = |\bar{S}|.
\]
\end{itemize}
Furthermore, a function satisfying the above properties can be constructed in $O(n^2)$ time.
\end{lemma}

\begin{remark}
The function $w$ described in Lemma~\ref{Weight} can be extended to a homomorphism $w: \mathcal{E}(G) \rightarrow \field{Z}$ because of property (i).
\end{remark}

\begin{proof}
Let $T=(V,E_T)$ be a spanning tree of $G$ and let $v$ be a root of $T$. If the lemma holds for $T$, then it clearly holds for $G$ by setting $w(a,e,b)=0$ for all $ab \in E \backslash E_T$.

Let $ab \in E_T$, and without loss of generality, assume that $b$ is a descendant of $a$ (that is $b$ is further from $v$ than $a$). Deleting $ab$ disconnects $T$ into two components $S_{ab}$ and $\bar{S}_{ab}$, where $\bar{S}_{ab}$ is the component not containing $v$ (and hence not containing $a$). Set $w(a,e,b)=-w(b,e,a)=|\bar{S}_{ab}|$. We do this for every $ab \in E_T$.

Clearly $w$ satisfies properties (i) and (ii), and furthermore, it is not hard to see that $w$ can be constructed in $O(n^2)$ time. It remains only to prove property (iii). We prove the claim by induction. Let $v_1, \ldots, v_r$ be the vertices of $T$ adjacent to $v$. Let $T_i=(V_i,E_i)$ be the subtree of $T$ formed from $v_i$ and its descendants for $i=1, \ldots, r$. Given $v \in S \subset V$ and $\bar{S}=V \backslash S$, let $S_i = S \cap V_i$ and $\bar{S_i}=\bar{S} \cap V_i$. Also, let 
\[ M = \{i: v_i \in S \}.
\]
Observe that
\[ \overrightarrow{[S,\bar{S}]} = \sum_{i=1}^r \overrightarrow{[S_i, \bar{S_i}]} + \sum_{i \not\in M} (v, vv_i, v_i).
\]
Now we have
\begin{align*}
w( \overrightarrow{[S,\bar{S}]} ) &= \sum_{i=1}^r w(\overrightarrow{[S_i, \bar{S_i}]}) + \sum_{i \not\in M} w(v, vv_i, v_i) \\
&= \sum_{i \in M}|\bar{S_i}| - \sum_{i \not\in M}|S_i| + \sum_{i \not\in M} w(v,vv_i,v_i) \:\:\:\text{ (induction)} \\
&= \sum_{i \in M}|\bar{S_i}| + \sum_{i \not\in M}(|V_i|-|S_i|) \\
&= \sum_{i=1}^r|\bar{S_i}| = |\bar{S}|.
\end{align*}
\hfill$\Box$
\end{proof}


Throughout, $w$ will be a homomorphism from $\mathcal{E}(G)$ to $\field{Z}$ satisfying the properties of Lemma~\ref{Weight}.

The domain for the function $d^f_G$ is the set of all cuts of $G$. It turns out that it is necessary to extend $d^f_G$ in a natural way to a function on $\mathcal{T}(G)$. Minimizing $d^f_G$ on $\mathcal{T}(G)$ will be equivalent to minimizing it on the set of cuts of $G$. Although $\mathcal{T}(G)$ is an infinite set, we shall eventually look to minimize $d^f_G$ over a suitable finite subset of $\mathcal{T}(G)$.  

For each $\phi \in \mathcal{T}(G)$, if $w(\phi) \not = 0$,  we define
\begin{equation}
 d^f_G(\phi) := \frac{|\phi|}{f(|w(\phi)|/n)}  
\end{equation}
where $n=|V|$; if $f(|w(\phi)|) = 0$, we define $d^f_G(\phi) = \infty$. Note that if $\phi = \overrightarrow{[S,\bar{S}]}$, then $|\phi| = |[S,\bar{S}]|$ and $|w(\phi)|$ is either $|S|$ or $|\bar{S}|$. Hence the function $d^f_G$ defined above extends the definition of $d^f_G$ given in the introduction.

Our next lemma shows that minimizing $d^f_G$ over $\mathcal{T}(G)$ is equivalent to minimizing $d^f_G$ over simple cuts $\overrightarrow{[S,\bar{S}]}$ of $G$. First some observations. For the cut space $\mathcal{T}(G)$ of $G=(V,E)$, we note that if $S$ is the union of disjoint subsets $S_1$ and $S_2$ of $V$, then 
\[ \overrightarrow{[S, \bar{S}]} = \overrightarrow{[S_1, \bar{S_1}]} + \overrightarrow{[S_2, \bar{S_2}]}.  
\] 
Note also that $\overrightarrow{[S, \bar{S}]} = -\overrightarrow{[\bar{S},S]}$. Thus every element of $\mathcal{T}(G)$ can be written as a positive integer linear combination of single vertex cuts $\overrightarrow{[\{v\},\bar{\{v\}}]}$. 

\begin{lemma} \label{Avecut}
Let $G=(V,E)$ be a graph. For every $\phi \in \mathcal{T}(G)$, there exists $S \subseteq V$ such that
\[ d^f_G(S,\bar{S}) \leq d^f_G(\phi),
\]
so in particular 
\[ \min_{\phi \in \mathcal{T}(G)} d^f_G(\phi) = d^f(G).
\]
\end{lemma}
The proof uses the same averaging argument used in the proof of Proposition~\ref{Concut}.
\begin{proof}
The lemma clearly holds if $d^f_G(\phi) = \infty$, so assume $d^f_G(\phi) < \infty$. Choose an integer $k$ such that 
\[ \phi = \sum_{v \in V} \lambda_v \overrightarrow{ [\{v\}, \bar{\{v\}}] },
\]
where $0 < \lambda_v \leq k$ for all $v \in V$. For each $i= 1, \ldots, k$, let
\[ S_i = \bigcup_{\lambda_v \geq i}\{v\}.
\]
Thus we have $S_1 \supseteq S_2 \supseteq \cdots \supseteq S_k$, and
\[ \phi = \sum_{i=1}^{k}  \overrightarrow{ [S_i,\bar{S_i}] }.
\]
Since the sets are nested, there is no cancelling of edges; hence
\[ |\phi| = \sum_{i=1}^{k} |\overrightarrow{[S_i,\bar{S_i}]}|.
\]
Using the fact that $w$ is a homomorphism together with the triangle inequality, we have
\[ |w(\phi)| = \Big| \sum_{i=1}^{k} w(\overrightarrow{[S_i,\bar{S_i}]}) \Big| \leq 
\sum_{i=1}^{k} |w(\overrightarrow{[S_i,\bar{S_i}]})|. 
\]
Now we have
\begin{align*}
d^f_G(\phi) = \frac{|\phi|}{f(|w(\phi)|/n)} &\geq 
\frac{ \sum_{i=1}^{k} |\overrightarrow{[S_i,\bar{S_i}]}| }{ f\big(\sum_{i=1}^{k}  |\frac{w(\overrightarrow{[S_i,\bar{S_i}]})}{n} |\big) } \\
&\geq 
\frac{ \sum_{i=1}^{k} |\overrightarrow{[S_i,\bar{S_i}]}| }{ \sum_{i=1}^{k} f(|\frac{w(\overrightarrow{[S_i,\bar{S_i}]})}{n} |) } &\text{Lemma \ref{Inequality} } \\
&\geq 
\frac{|\overrightarrow{[S_{i'},\bar{S_{i'}}]}|}{f(|\frac{w(\overrightarrow{[S_{i'},\bar{S_{i'}}]})}{n} |)} &\text{ for some $i'$ by Proposition~\ref{Ave} } \\
&= 
d^f_G(S_{i'},\bar{S_{i'}}).
\end{align*}
\hfill$\Box$
\end{proof}

We specify a \emph{walk} $w$ of a graph $G=(V,E)$ by giving an alternating sequence of vertices and edges $w=(x_1,e_1,x_2,e_2, \ldots, x_{k-1},e_{k-1},x_k)$, where $(x_i,e_i,x_{i+1}) \in \overrightarrow{E}$ for all $i$. We write $|w|$ for the number of edges traversed in $w$ (which in this case is $k-1$). If $x_1=x_k$ then $w$ is called a \emph{closed walk}. If all edges of a closed walk $w$ are distinct, then $w$ is called a called a \emph{circuit} of $G$. If all the vertices of a closed walk $w$ are distinct (except $x_1=x_k$), then $w$ is called a \emph{cycle}. A walk in which all vertices are distinct is called a \emph{path}.  We write $\overrightarrow{w}$ for the element of $\mathcal{E}(G)$ given by
\[ \overrightarrow{w} = \sum_{i=1}^{k-1} (x_i, e_i, x_{i+1});
\]
we refer to $\overrightarrow{w}$ as an \emph{oriented} walk, (circuit, etc). Note that for a walk $w$, we have $|w| \geq |\overrightarrow{w}|$ with equality when $w$ is a circuit.

The \emph{cycle space} $\mathcal{C}(G)$ is the subgroup of $\mathcal{E}(G)$ (redundantly) generated by the oriented cycles $\overrightarrow{c}$ of $G$. Note that $\mathcal{C}(G)$ contains all oriented closed walks of $G$.

We now turn our attention to graphs embedded on closed orientable surfaces. Formally, a surface is a compact connected topological space in which every point of the surface has an open neighbourhood homeomorphic to $\field{R}^2$ or the closed halfplane $\{(x,y) \in \field{R}^2: y \geq 0\}$. The set of points having halfplane open neighbourhoods is called the boundary of the surface. Every component of the boundary is homeomorphic to the circle $S^1$. A closed surface is one without boundary. A surface is called orientable if it does not contain a subset (with the subset topology) homeomorphic to the m{\"o}bius band.

The genus of a connected, orientable surface is the maximum number of cuttings along non-intersecting closed simple curves that can be made without disconnecting the surface. It is well known from the classification of surfaces that every closed orientable surface of genus $g$ is homeomorphic to a sphere with $g$ handles. Thus we can think of every such surface as embedded in $\field{R}^3$. Informally, a graph $G$ can be embedded on a surface $\Sigma$ if $G$ can be drawn on $\Sigma$ in such a way that no edge crosses a vertex or another edge, except possibly at its end points. For example, all planar graphs can be embedded on the sphere.

More formally, we have the following definitions. A (topological) path in $\Sigma$ is a continuous function $\gamma:[0,1] \rightarrow \Sigma$, where $\gamma(0)$ and $\gamma(1)$ are called the end points of the path. If $\gamma(0)=\gamma(1)$ then $\gamma$ is called a (topological) cycle. Paths and cycles are referred to collectively as curves. A curve is simple if it is injective, except in the case of a cycle where we permit $\gamma(0)=\gamma(1)$. We often do not distinguish between curves and their images in $\Sigma$. 

Throughout, we shall consider embeddings of multigraphs with loops on orientable surfaces. (Although $G$ has no loops, its dual graph may have loops as we shall discuss later.) An embedding of a graph $G$ in $\Sigma$ maps vertices $v$ to distinct points $\psi(v)$ of $\Sigma$ and maps edges (resp.\ loops) $e=ab$ to simple topological paths (resp.\ cycles) $\gamma_e$ of $\Sigma$: two such paths may intersect (if at all) only at their end points, and the end points of $\gamma_e$ are precisely $\psi(a)$ and $\psi(b)$. Oriented edges $\overrightarrow{e}=(a,e,b)$ are embedded by simple paths $\gamma_{\overrightarrow{e}}$, where we insist that $\gamma_{\overrightarrow{e}}(0)=\psi(a)$ and $\gamma_{\overrightarrow{e}}(1)=\psi(b)$. We often abuse terminology and notation by identifying vertices, edges, and walks of $G$ with their images in the embedding of $G$ on $\Sigma$. 

From an algorithmic point of view, one can use \emph{rotation systems} to input or output embeddings of graphs on surfaces. We do not define rotation systems here because we shall only use graph embeddings indirectly when applying existing algorithms to our problem; instead we refer the reader to \cite{Tho}.

Throughout, we shall only consider cellular embeddings. An embedding of $G$ on $\Sigma$ is called \emph{cellular} if removing the image of $G$ from $\Sigma$ leaves a set of topological disks called the \emph{faces} of $G$. The \emph{genus} of a graph $G$ is defined to be the smallest integer $g$ such that $G$ can be embedded on a closed orientable surface of genus $g$. If $G$ is a graph of genus $g$, then every embedding of $G$ on a closed orientable surface of genus $g$ is cellular (Proposition 3.4.1 \cite{Tho}), and for fixed $g$, such an embedding can be found in linear time \cite{Moh}. Euler's Theorem gives the following relationship between the number of vertices $n$, the number of edges $m$, the number of faces $\ell$, and the number of boundary components $b$ in a cellular embedding of a graph $G$ on a surface $\Sigma$ of genus $g$:
\[n - m + \ell + b = 2-2g.
\]

We now define the \emph{boundary space} of a graph $G$ embedded on $\Sigma$. Each oriented edge $\overrightarrow{e}$ of $G$ separates two (possibly equal) faces of $G$ denoted $left(\overrightarrow{e})$ and $right(\overrightarrow{e})$. (The notion of left and right with respect to an oriented edge is well defined for orientable surfaces.) For a face $F$, the oriented edges $\overrightarrow{e}$ for which $left(\overrightarrow{e}) = F$ taken in order (with appropriate intervenng vertices) form a closed walk $f$ around $F$, which we call the facial walk of $F$.  Let $F_1, \ldots, F_{\ell}$ be the faces of the embedding and let $f_i$ be the facial walk of $F_i$. The oriented facial walk $\overrightarrow{f_i}$ of $F_i$ is given by 
\[ \overrightarrow{f_i} = \sum_{\overrightarrow{e}: \: left(\overrightarrow{e})=F_i} \overrightarrow{e}. 
\]
Notice that $f_i$ may contain two oppositely oriented edges, but such edges cancel in $\overrightarrow{f_i}$, leaving the sum of oriented edges that form the boundary of $F_i$. The boundary space $\mathcal{B}(G, \Sigma)$ is defined to be the group generated by $\overrightarrow{f_1}, \ldots, \overrightarrow{f_{\ell}}$ and is easily seen to be a subgroup of $\mathcal{C}(G)$.
Note that the boundary space of $G$, in contrast to the cycle space and cut space, depends on the embedding of $G$. 



The geometric dual of a graph $G$ cellularly embedded on $\Sigma$ is denoted by $D(G)=(V',E')$ and is the graph (with cellular embedding on $\Sigma$) constructed from $G$ as follows. For each face $F$ of $G$, a vertex $D(F)$ is placed inside $F$: these are the vertices of $D(G)$. Each edge $e$ of $G$ has a corresponding edge $D(e)$ in $D(G)$: $D(e)=D(F_1)D(F_2)$, where $e$ is the edge separating the (possibly indistinct) faces $F_1$ and $F_2$ (and $D(e)$ crosses $e$ and no other edge of $G$ in the embedding of $D(G)$). Note that $G$ (with its embedding) is a dual of $D(G)$ (with its embedding). Therefore, each vertex $v$ of $G$ corresponds to a face $D(v)$ of $D(G)$. Thus $D$ maps vertices, edges, and faces of $G$ bijectively to faces, edges, and vertices of $D(G)$ respectively. We extend $D$ to map oriented edges of $G$ to oriented edges of $D(G)$ as follows. Given an oriented edge $\overrightarrow{e}$, we set $D(\overrightarrow{e}) = (D(left(\overrightarrow{e})),D(e),D(right(\overrightarrow{e})))$. This, however, reverses the sense of left and right so that $D(D(\overrightarrow{e}))= -\overrightarrow{e}$. 
\begin{remark}
Although $G$ is a loopless graph, $D(G)$ may not be. Nonetheless, all notions introduced so far carry through naturally for loops of embedded graphs. In particular, a loop $e$ of an embedded graph $G$ has two orientations $\overrightarrow{e}$ and  $\overleftarrow{e}$, although we cannot specify which is which by the order of the end vertices. If $D(e)$ is a loop for some edge $e$ of $G$, then $D(\overrightarrow{e})$ should cross $\overrightarrow{e}$ from left to right.
\end{remark}

Since $D$ bijectively maps edges of $G$ to edges of $D(G)$, we see that $D$ can be extended to an isomorphism $D: \mathcal{E}(G) \rightarrow \mathcal{E}(D(G))$. We have the following well-known correspondence.

\begin{theorem}
The restriction of $D$ to $\mathcal{T}(G)$ gives an isomorphism $\mathcal{T}(G) \rightarrow \mathcal{B}(D(G))$.
\end{theorem}
\begin{proof}
Note that for a vertex $v$ of $G$, and $f_v$ the facial walk around $D(v)$, we have
\[ D(\overrightarrow{[\{v\}, \bar{\{v\}}]}) = -\overrightarrow{f_v}. 
\]
This defines a bijective correspondence.
\hfill$\Box$
\end{proof}

Rather than working with $\mathcal{T}(G)$, we can work instead with $\mathcal{B}(D(G))$ using the isomorphism $D$. Since all boundaries are sums of oriented cycles, we can use shortest-path algorithms to find shortest boundaries in $D(G)$, which if done suitably, can give us $f$-sparsest cuts in $G$.

For $\sigma \in \mathcal{B}(D(G))$, define $\hat{w}(\sigma):=w(D^{-1}(\sigma))$. Notice that for all $\sigma \in \mathcal{B}(D(G))$, we have $|D^{-1}(\sigma)| = |\sigma|$. Thus, defining 
\begin{equation}
\hat{d}^f_{D(G)}(\sigma) = \frac{|\sigma|}{f(|\hat{w}(\sigma)|)},
\end{equation}
we have that
\[ \min_{\sigma \in \mathcal{B}(D(G))} \hat{d}^f_{D(G)}(\sigma) = \min_{\phi \in \mathcal{T}(G)} d^f_G(\phi) = d^f(G).
\]

We now set about trying to minimize $\hat{d}^f_{D(G)}$. Our next lemma says that when minimizing $\hat{d}^f_{D(G)}$, we can restrict attention to elements of $\mathcal{B}(D(G))$ that are the sum of at most $g+1$ oriented circuits (where $g$ is the genus of $G$).

 
First a proposition. 

\begin{proposition} \label{Circuits}
Let $G(V,E)$ be a graph cellularly embedded on $\Sigma$. For every cut $[S,\bar{S}]$ of $G$ there exist vertex disjoint circuits $w_1, \ldots, w_r$ such that 
\[ \sigma = D(\overrightarrow{[S,\bar{S}]}) = \overrightarrow{w_1} + \cdots + \overrightarrow{w_r}
\] 
and
\[ |\sigma| = |\overrightarrow{w_1}| + \cdots + |\overrightarrow{w_r}|. 
\]
\end{proposition}
\begin{proof}
\begin{proof} 
We know $\sigma = D(\overrightarrow{[S,\bar{S}]})$ is an element of $\mathcal{B}(D(G)) \subseteq \mathcal{C}(D(G))$ in which no edge occurs more than once. Thus we can write 
\[ \sigma = \sum_{i=1}^s \overrightarrow{c_i} = \sum_{\overrightarrow{e} \in \overrightarrow{E}} \lambda_{\overrightarrow{e}}\overrightarrow{e},
\]
where each $\overrightarrow{c_i}$ is an oriented cycle and each $\lambda_{\overrightarrow{e}} \in \{0,1\}$. The number of oriented edges of $\sigma$ entering and exiting any given vertex must be equal since this is the case for any oriented cycle and remains the case after cancellation of oriented edges when summing oriented cycles. Thus we can walk around $D(G)$ using the oriented edges of $\{\overrightarrow{e}: \lambda_{\overrightarrow{e}}=1\}$ to form disjoint closed circuits $w_1, \ldots, w_r$, where $\sigma = \overrightarrow{w_1} + \cdots + \overrightarrow{w_r}$.
Furthermore, there is no cancellation of edges when we add these disjoint oriented circuits together; hence
\[ |\sigma| = |\overrightarrow{w_1}| + \cdots + |\overrightarrow{w_r}|.
\]
\hfill$\Box$
\end{proof}
\end{proof}

Recall that for each vertex $v$ of $G$, $D(v)$ is a face (a topological disk in $\Sigma$) of $D(G)$ and $f_v$ is the facial walk of $D(v)$. The \emph{closed face} $D^*(v)$ is a closed disk which has the facial walk $f_v$ embedded along its boundary. Every $\overrightarrow{e} \in \overrightarrow{E}$ occurs on the boundary of some $D^*(v)$; thus each edge of $G$ is either embedded on two distinct closed faces or is embedded twice on the same closed face. Note that $\Sigma$ can be constructed by gluing these closed faces together along common edges of $G$ (respecting the orientation). We can now prove our lemma.

\begin{lemma} \label{g+1}
Suppose $G=(V,E)$ is a graph cellularly embedded on a surface $\Sigma$ of genus $g$. There exists $\sigma \in \mathcal{B}(D(G))$ that minimizes $\hat{d}^f_{D(G)}$ such that
\[ \sigma = \overrightarrow{w_1} + \cdots + \overrightarrow{w_r},
\] 
where $\overrightarrow{w_1}, \ldots, \overrightarrow{w_r}$ are disjoint oriented circuits of $D(G)$ and $r \leq g+1$. 
Furthermore $m \geq |\sigma| = |\overrightarrow{w_1}| + \cdots + |\overrightarrow{w_r}|$, where $m=|E|$.
\end{lemma}
\begin{proof}
By Proposition~\ref{Concut}, there exists an $f$-sparsest cut $[S,\bar{S}]$ of $G$ for which $G[S]$ and $G[\bar{S}]$ are both connected. Thus $\overrightarrow{[S,\bar{S}]} \in \mathcal{T}(G)$ minimizes $d^f_G$ (over all elements of $\mathcal{T}(G)$ by Lemma~\ref{Avecut}), and so $\sigma = D(\overrightarrow{[S,\bar{S}]})$ minimizes $\hat{d}^f_{D(G)}$ (over all elements of $\mathcal{B}(D(G))$).  By Proposition~\ref{Circuits}, we know that we can write $\sigma = D(\overrightarrow{[S,\bar{S}]}) = \overrightarrow{w_1} + \cdots + \overrightarrow{w_r}$, where $m \geq |\sigma| = |\overrightarrow{w_1}| + \cdots + |\overrightarrow{w_r}|$. It remains to show that $r \leq g+1$.

Let $(H_S,\Sigma_S)$ be obtained from $(G,\Sigma)$ as follows. Take the set of closed faces $\{D^*(v):v \in S\}$ and for each edge $v_1v_2 \in G[S]$, glue $D^*(v_1)$ and $D^*(v_2)$ together along their common edge $D(v_1v_2)$, respecting orientation, to form $\Sigma_S$. Since $G[S]$ is connected, $\Sigma_S$ is a surface with boundary. The facial cycles of the glued faces now give a graph $H_S$ (which is a subgraph of $D(G)$) embedded on $\Sigma_S$. Observe that each edge of $D(G)$ occurs at most once as an edge of $H_S$ and the edges of $w_1, \ldots, w_r$ form the boundary of $\Sigma_S$. Disjoint edges of $D(G)$, when they occur in $H_S$, remain disjoint; hence, since $w_1, \ldots, w_r$ are disjoint, $\Sigma_S$ must have at least $r$ boundary components. We form $(H_{\bar{S}},\Sigma_{\bar{S}})$ analogously and by symmetry it has the same properties described above. Note that $\Sigma$ can be formed from $\Sigma_S$ and $\Sigma_{\bar{S}}$ by gluing them together suitably along their boundaries. Note also that the edges of $H_S$ are precisely the duals of the edges of $G$ incident with $S$: similarly for $H_{\bar{S}}$. We apply Euler's theorem to the two embeddings.

Let $g$, $g_S$ and $g_{\bar{S}}$ be the genii of $\Sigma$, $\Sigma_S$, and $\Sigma_{\bar{S}}$ respectively. Let $n$, $m$, $\ell$, and $b=0$ denote the number of vertices, edges, faces, and boundary components for the embedding of $G$ on $\Sigma$. Let $n_S$, $m_S$, $\ell_S$, and $b_S$ denote the numbers of vertices, edges, faces, and boundary components respectively for the embedding of $H_S$ on $\Sigma_S$, and analogously for $\bar{S}$. Finally, let $n_b$ and $m_b$ denote the numbers of vertices and edges that occur on the boundaries of $\Sigma_S$ and $\Sigma_{\bar{S}}$. We have
\begin{align*}
n &= n_1 + n_2 - n_b ,\\
m &= m_1 + m_2 - m_b,  \\
\ell &= \ell_1 + \ell_2, \\
n_b &= m_b \:\:\: \text{(since boundaries consist of disjoint cycles)} \\
b_S, b_{\bar{S}} &\geq r
\end{align*}
Using Euler's formula and the first three equalities, we have
\[
2- 2g = n - m + \ell + b = (2-2g_S) + (2-2g_{\bar{S}}) + m_b - n_b - b_S - b_{\bar{S}}.
\]
Rearranging, and using the fourth and fifth statements above, we have
\[
g+1 = g_S + g_{\bar{S}} + \frac{1}{2}\Big( (n_b - m_b) + (b_S + b_{\bar{S}}) \Big) \geq r.
\] 
\hfill$\Box$
\end{proof}

We can easily use shortest-path algorithms to find shortest cycles in a graph. However in this situation, we are required to find a shortest boundary (loosely speaking). We require a simple way of testing whether a cycle is a boundary. This is accomplished using the ideas of homology. We only require the following fact about the homology of graphs on surfaces. If $G$ is a graph cellularly embedded on an orientable surface $\Sigma$ of genus $g$, then the quotient group $\mathcal{C}(G) / \mathcal{B}(G)$ is isomorphic to $\field{Z}^{2g}$; this follows easily from standard results on cellular homology.


For an $n$-vertex graph $G$ cellularly embedded on a closed orientable surface $\Sigma$ of genus $g$, a \emph{system of loops} is a set of cycles of $G$ through a common vertex such that cutting $\Sigma$ along these cycles gives a topological disk. By Euler's formula, every system of loops must consist of $2g$ cycles. Erickson and Whittlesey \cite{Eric} give a greedy algorithm, which, given a cellular embedding of $G$ on $\Sigma$, finds a system of loops $c(1), \ldots, c(2g)$ in $O(n\log{n} + gn)$ time.  

We define a homomorphism $\Theta_i: \mathcal{C}(D(G)) \rightarrow \field{Z}$, where, for every $\sigma \in \mathcal{C}(D(G))$, the integer $\Theta_i(\sigma )$ measures the net number of times $\sigma$ crosses $c(i)$ (here sign indicates the direction of crossings). Let us define $\Theta_i$ formally.

For each oriented edge $\overrightarrow{e}$ of $G$ and each oriented edge $\overrightarrow{e^*}$ of $D(G)$, define
\[ \Theta_{\overrightarrow{e}}(\overrightarrow{e^*}) = 
\begin{cases}
1 &\text{ if } D(\overrightarrow{e}) = \overrightarrow{e^*}; \\
-1 &\text{ if } D(\overrightarrow{e}) = -\overrightarrow{e^*}; \\
0 &\text{ otherwise.}
\end{cases}
\]
For every $\phi = \sum_i\lambda_i \overrightarrow{e_i} \in \mathcal{E}(G)$ and every $\sigma = \sum_j\mu_j \overrightarrow{e^*_j} \in \mathcal{E}(D(G))$, we define
\[  \Theta_{\phi}(\sigma ) = \sum_i \sum_j \lambda_i \mu_j \Theta_{\overrightarrow{e_i}}(\overrightarrow{e_j^*}),
\]
which counts the directed number of times $\phi$ and $\pi$ cross each other. It is easy to check that the above is well defined. We write $\Theta_i$ as a shorthand for $\Theta_{\overrightarrow{c(i)}}$. Finally, define $\Theta: \mathcal{E}(D(G)) \rightarrow \field{Z}^{2g}$ by setting
\[ \Theta(\sigma) = \big(\Theta_1(\sigma), \ldots, \Theta_{2g}(\sigma) \big)
\]
for every $\sigma \in \mathcal{E}(D(G))$. 

We have the following proposition which effectively says that any cycle of $D(G)$ that intersects each $\overrightarrow{c(i)}$ a net number of zero times is a boundary. It is very much what we expect from the properties of homology, but we give the details for completeness.
\begin{proposition} \label{Theta}
The function $\Theta$ defined above is a surjective homomorphism from $\mathcal{C}(D(G))$ to $\field{Z}^{2g}$ whose kernel is $\mathcal{B}(D(G))$. 
\end{proposition}
\begin{proof} 
The fact that $\Theta$ is well defined and a homomorphism is easy to check. 

To see that $\Theta$ is surjective, we must find for each $j=1, \ldots, 2g$, a cycle $c^*(j) \in \mathcal{C}(D(G))$ such that $\Theta(c^*(j)) = \pm \vec{u}_j$, where $\vec{u}_j$ is the vector that has $1$ in the jth component and $0$'s elsewhere. 
Consider the embedding of $G$ on the (closed) topological disk $T$ formed by ungluing $\Sigma$ along the cycles $c(1), \ldots, c(2g)$; thus each unglued edge is embedded twice along the boundary of $T$. Note that each $c(j)$ has at least one edge $e_j$ that does not belong to any of the other $c(i)$'s: indeed, if all the edges of $c(j)$ belonged to other cycles, then $c(j)$ would be redundant and we could cut the surface into a topological disk with fewer than $2g$ cycles.

Pick points $x$ and $\bar{x}$ on the boundary of $T$ that lie in the interior of the two embeddings of $e_j$. Let $\gamma$ be a topological path in $T$ from $x$ to $\bar{x}$ that is not incident with any other points along the boundary of $T$ and not incident with any vertices of $G$. It is clear that such a path exists. Then $\gamma$ corresponds to a topological cycle in $\Sigma$ that crosses the cycle $\overrightarrow{c(j)}$ exactly once and does not cross any of the other $c(i)$'s. By listing the alternating sequence of faces and edges of $G$ that $\gamma$ crosses, we obtain a cycle $c^*(j)$ in $D(G)$ and hence an oriented cycle $\overrightarrow{c^*(j)} \in \mathcal{C}(D(G))$  that crosses $\overrightarrow{c(j)}$ once but does not cross any of the other $c(i)$'s. Thus, we have
\[ \Theta(\overrightarrow{c^*(j)}) = \pm \vec{u}_j,
\] 
showing that $\Theta$ is surjective.

Finally we show that $\ker(\Theta) = \mathcal{B}(D(G))$. For any oriented cycle $\overrightarrow{C}$ and any oriented facial walk $\overrightarrow{f}$ of $D(G)$, we have 
\[ \Theta_{\overrightarrow{C}}(\overrightarrow{f_i})=0
\]
since $\overrightarrow{C}$ enters and exits the face $F_i$ an equal number of times so that all contributions cancel. Thus, we have that $\Theta(\overrightarrow{f})= \vec{0}$ for all oriented facial walks $\overrightarrow{f}$ of $D(G)$. Therefore $\Theta(\phi)=\vec{0}$ for all $\phi \in \mathcal{B}(D(G))$ showing that $\mathcal{B}(D(G)) \subseteq \ker(\Theta)$. Now we have
\[ \field{Z}^{2g} \cong \frac{\mathcal{C}(D(G))}{\ker(\phi)} \cong 
   \frac{\mathcal{C}(D(G))/\mathcal{B}(D(G))}{\ker(\phi)/\mathcal{B}(D(G))} \cong 
   \frac{\field{Z}^{2g}}{\ker(\phi)/\mathcal{B}(D(G))}, 
\]
where the first isomorphism is from the First Isomorphism Theorem, the second is from the Third Isomorphism Theorem, and the third is given by the fact about homology groups given earlier. We deduce that $\ker(\phi)/\mathcal{B}(D(G))$ must be the trivial group (by standard properties of finitely generated abelian groups) and that therefore we must have $\ker(\phi)=\mathcal{B}(D(G))$.
\hfill$\Box$
\end{proof}

For convenience, we combine some of the results we have so far to give the following proposition.

\begin{proposition} \label{Final}
There exists an element $\sigma \in \mathcal{B}(D(G))$ that minimizes $\hat{d}^f_{D(G)}$ and satisfies the following properties.
\begin{itemize}
\item[(a)] We can write $\sigma$  as  $\overrightarrow{w_1} + \cdots + \overrightarrow{w_r}$ where the $\overrightarrow{w_i}$ are oriented circuits in $D(G)$ and $r \leq g+1$.
\item[(b)] We have $|\sigma| = |\overrightarrow{w_1}| + \cdots + |\overrightarrow{w_r}|$ with $|\overrightarrow{w_i}| \leq m$ for all $i$.
\item[(c)] $\Theta(\sigma) = \Theta(\overrightarrow{w_1}) + \cdots + \Theta(\overrightarrow{w_r})= \vec{0}$. 
\end{itemize}
\end{proposition}
\begin{proof}
Statements (a) and (b) follow from Lemma~\ref{g+1}, and statement (c) follows from Proposition~\ref{Theta}.
\hfill$\Box$
\end{proof}

Next we define a \emph{covering graph} of $D(G)$ in which certain shortest paths will correspond to elements of $\mathcal{C}(D(G))$ whose sum will minimize $\hat{d}^f_{D(G)}$.

We construct an infinite multigraph $H=(V_H,E_H)$ from $D(G)$ as follows. We set $V_H = V' \times \field{Z} \times \field{Z}^{2g}$, where $V'$ is the vertex set of $D(G)$. For convenience, we describe the oriented edges of $H$ before describing its edges. For each oriented edge $\overrightarrow{e} = (u_1, e, u_2)$ of $D(G)$ and each $(k,\vec{v}) \in \field{Z} \times \field{Z}^{2g}$, we have an oriented edge of $H$, denoted $(\overrightarrow{e},k,\vec{v})^{\ast}$, from
\[ (u_1, k, \vec{v}) \:\:\: \text{ to } \:\:\: (u_2, k + \hat{w}(\overrightarrow{e}), \vec{v} + \Theta(\overrightarrow{e}) ).
\]
The same edge oriented oppositely is given by $(\overleftarrow{e},k+ \hat{w}(\overrightarrow{e}),\vec{v} + \Theta(\overrightarrow{e}))^{\ast}$. We write $(\overrightarrow{e},k,\vec{v})$ for the edge of $H$ corresponding to the oriented edge $(\overrightarrow{e},k,\vec{v})^{\ast}$; thus every edge of $H$ has two labels.

Walks of $D(G)$ naturally correspond to walks of $H$ as follows. Let $w=(u_0, e_1, u_1, \ldots, e_{t-1}, u_t)$ be a walk of $D(G)$. Let $w_i = (u_0, e_1, u_1, \ldots, e_{i-1}, u_i)$ be the same walk up to the $i$th vertex, and let $\overrightarrow{e_i} = (u_{i-1},e_i, u_i)$. Let $H(w)$ be the walk in $H$ given by $H(w) = (u^H_0, e^H_1, u^H_1, \ldots, e^H_{t-1}, u^H_t)$, where $u^H_0=(u,0,\vec{0})$ and
\[ u^H_i = (u_i, \hat{w}(\overrightarrow{w_i}), \Theta (\overrightarrow{w_i})) 
\:\:\: \text{ and } \:\:\: e^H_i = (\overrightarrow{e_i}, \hat{w}(\overrightarrow{w_{i-1}}), \Theta(\overrightarrow{w_{i-1}})).
\]
It is easy to check that $w \mapsto H(w)$ is a bijective correspondence between walks of $D(G)$ and walks of $H$ that start at $(u,0,\vec{0})$ for some $u \in V'$. Furthermore, $w$ is a walk of $D(G)$ from $u$ to $u'$ and satisfies $\hat{w}(\overrightarrow{w})=k$ and $\Theta(\overrightarrow{w})=\vec{v}$ if and only if $H(w)$ is a walk of $H$ from $(u,0,\vec{0})$ to $(u', k, \vec{v})$.

Defining $V_H^* = V' \times \{-mn, \ldots, mn \} \times \{-m, \ldots, m\}^{2g} \subset V_H$, let $H^*$ be the finite graph induced by $H$ on $V_H^*$.
For $(u,k,\vec{v}) \in V_H^*$, define $p(u,k,\vec{v})$ to be a shortest path in $H^*$ from $(u,0,\vec{0})$ to $(u,k,\vec{v})$ (if it exists); this path corresponds to a closed walk in $D(G)$. For fixed $k$ and $\vec{v}$, let $p(k,\vec{v})$ be the path of minimum length in $\{p(u,k,\vec{v}): u \in U\}$ (if it exists). Let $w(k,\vec{v})$ be the (closed) walk of $D(G)$ corresponding to the path $p(k,\vec{v})$ in $H^*$. We have the following lemma. 

\begin{lemma} \label{Less}
Suppose $w$ is a circuit of $D(G)$ such that $\hat{w}(\overrightarrow{w})=k$ and $\Theta(\overrightarrow{w})=\vec{v}$ and $|\overrightarrow{w}| \leq m$. Then
\[  |\overrightarrow{w}| \geq |\overrightarrow{w(k,\vec{v})}|.
\]  
\end{lemma}
\begin{proof}
Suppose the circuit $w$ starts at $u \in U$.
We shall show that $H(w)$ is a walk on $H^{\ast}$ (from $(u,0,\vec{0})$ to $(u,k,\vec{v})$). Then using the fact that $w$ is a circuit and consequently $|\overrightarrow{w}|=|w|$, we have
\[ |\overrightarrow{w}| = |w| = |H(w)| \geq |p(u,k,\vec{v})| \geq |p(k,\vec{v})| = |w(k,\vec{v})| 
\geq |\overrightarrow{w(k,\vec{v})}|,
\]
proving the lemma.

To see that $H(w)$ is a walk on $H^*$, we note first that $|\overrightarrow{w}| = |w| \leq m$. Let $w_i$ be the same walk as $w$ up to the $i$th vertex. Then by property (ii) from Lemma~\ref{Weight}, we have that $|\hat{w}(\overrightarrow{w_i})| \leq |\overrightarrow{w_i}|n \leq mn$. Also, since $\overrightarrow{w} \leq m$, it is not hard to see that $\Theta(\overrightarrow{w_i}) \in \{-m, \ldots, m\}^{2g}$. This shows that the walk $H(w)$ never leaves $V_H^*$, as required.

\hfill$\Box$
\end{proof}
Let $W$ be the set of all the $\overrightarrow{w(k, \vec{v})}$. Let  
\[ X = \{\overrightarrow{w_1} + \cdots + \overrightarrow{w_r}: \overrightarrow{w_i} \in W \: \forall i \:\text{ and }\: r \leq g+1 \},
\] 
and let $Y = \{ \sigma \in X: \Theta(\sigma) = \vec{0}\}$. We have the following corollary.
\begin{corollary} \label{Last}
There exists an element of $Y$ that minimizes $\hat{d}^f_{D(G)}$.
\end{corollary}
\begin{proof}
Clearly every element of $Y$ is in $\mathcal{B}(D(G))$. Let $\sigma \in \mathcal{B}(D(G))$ minimize $\hat{d}^f_{D(G)}$ and satisfy the three conditions of Proposition~\ref{Final}; in particular
\[ \sigma = \overrightarrow{w_1} + \cdots + \overrightarrow{w_r}
\] 
for $r \leq g+1$, where each $|\overrightarrow{w_i}| \leq m$, and where $\Theta(\sigma) = \vec{0}$. For each $i$, let $\Theta(\overrightarrow{w_i}) = \vec{v}_i$ and let $\hat{w}(\overrightarrow{w_i}) = k_i$. Then by Lemma~\ref{Less} we have that 
\[ |\overrightarrow{w(k_i,\vec{v}_i)}| \leq |\overrightarrow{w_i}|.
\]
Let $\sigma^* \in X$ be given by 
\[ \sigma^* = \overrightarrow{w(k_1,\vec{v}_1)} + \cdots + \overrightarrow{w(k_r,\vec{v}_r)}.
\] 
We have that
\[ |\sigma| = |\overrightarrow{w_1}| + \cdots + |\overrightarrow{w_r}|
           \geq |\overrightarrow{w(k_1,\vec{v}_1)}| + \cdots + |\overrightarrow{w(k_r,\vec{v}_r)}|
  \geq |\sigma^*|,
\] 
that $\Theta(\sigma^*)=\Theta(\sigma)=\vec{0}$, and that $\hat{w}(\sigma^*)=\hat{w}(\sigma)$. Thus $\sigma^* \in Y$ and
\[ \hat{d}^f_G(\sigma^*) = \frac{|\sigma^*|}{f(|\hat{w}(\sigma^* )| )} \leq \frac{|\sigma|}{f(|\hat{w}(\sigma )| )}
= \hat{d}^f_G(\sigma),
\] 
showing that $\sigma^* \in Y$ minimizes $\hat{d}^f_G$.
\hfill$\Box$
\end{proof}
  
We now have all the ingredients to present our algorithm and to prove its correctness.

\section{The Algorithm}

In this section, we present the basic steps of our algorithm and compute its running time. In order to keep the presentation simple, we do not optimize the running time. Our algorithm runs in time $O(n^{2g^2 + 4g + 7})$, and it seems unlikely that our methods can give an $n^{o(g^2)}$-time algorithm without significant modification.  

Let $f: [0, \frac{1}{2}] \rightarrow [0, \infty)$ be a fixed concave, increasing function that is computable in polynomial time on the rationals, and let $g$ be a fixed non-negative integer. The input for our algorithm is an $n$-vertex undirected multigraph $G$ of genus $g$. Since $n=|V|$ then $m=|E|=O(n)$ from Euler's formula. Using the result of Mohar \cite{Moh} mentioned earlier, we can find an embedding of $G$ on a surface $\Sigma$ of genus $g$ in $O(n)$ time. From the embedding we can construct (in $O(n)$ time) the dual graph $D(G) = (V',E')$ together with the function $D$ which maps each oriented edge $\overrightarrow{e} \in \overrightarrow{E}$ of $G$ to its dual $D(\overrightarrow{e}) \in  \overrightarrow{E'}$ in $D(G)$. We have $|E'|=|E|=O(n)$, and from Euler's formula, we have $|V'|=O(n)$. By the result of Erickson and Whittlesey \cite{Eric} mentioned earlier, we can find a system of loops $\overrightarrow{c(1)}, \ldots, \overrightarrow{c(2g)}$ of $G$ in time $O(n)$.  At this point the algorithm no longer requires the embedding of $G$.
 
 Next we construct and store the (restricted) functions $\hat{w}: \overrightarrow{E'} \rightarrow \field{Z}$ and $\Theta : \overrightarrow{E'} \rightarrow \field{Z}^{2g}$. The computation and storage of these functions imposes an insignificant time cost in the final analysis, so any crude bound on the running time is sufficient.
 Recall that a function $w : \overrightarrow{E} \rightarrow \field{Z}$ satisfying the properties of Lemma~\ref{Weight} can be constructed in $O(n^2)$ time. Then $\hat{w} = w \circ D^{-1}$ can also be constructed in $O(n^2)$ time. In order to construct $\Theta$, observe that we can construct $\Theta_{\overrightarrow{e}}: \overrightarrow{E'} \rightarrow \field{Z}$ in $O(n)$ time. Now $\Theta_i$ can be computed in $O(n^2)$ time since if $\overrightarrow{c(i)} = \overrightarrow{e_1} + \cdots + \overrightarrow{e_k}$ (where $k \leq n$, then 
\[
\Theta_i = \Theta_{\overrightarrow{c(i)}} = \Theta_{\overrightarrow{e_1}} + \cdots + \Theta_{\overrightarrow{e_k}}.  
\]
Thus $\Theta = (\Theta_1, \ldots, \Theta_{2g})$ can be computed in $2gO(n^2) = O(n^2)$ time.

From the (restricted) functions $\hat{w}$ and $\Theta$, we can construct the graph $H^*$ directly from its definition (given in the previous section). Observe that $H^*$ has $O(n) \cdot (2mn+1) \cdot (2m+1)^{2g} = O(n^{2g+3})$ vertices and $O(n^{2g+3})$ edges. Thus it takes $O(n^{2g+3})$ time to construct $H^*$. For each $(u,k,\vec{v}) \in V_{H^*}$, we compute and store the shortest path $p(u,k,\vec{v})$ from $(u,0,\vec{0})$ to $(u,k,\vec{v})$ in $H^*$. Finding each shortest path requires $O(|V_{H^*}| \log(|V_{H^*}|)) = O(n^{2g+4})$ time using Dijkstra's algorithm, and so, computing all the $p(u,k,\vec{v})$ requires $|V_{H^*}|O(n^{2g+4}) = O(n^{4g+7})$ time. For each fixed $k \in \{-mn, \ldots, mn \}$ and  $\vec{v} \in \{-n, \ldots, n\}^{2g}$, we compute and store $p(k,\vec{v})$, the path of minimum length amongst the $p(u,k,\vec{v})$, and we use
$p(k,\vec{v})$ to compute and store $w(k,\vec{v})$ and $\overrightarrow{w(k,\vec{v})}$ (recall that $w(k,\vec{v})$ is the closed walk in $D(G)$ corresponding to $p(k,\vec{v})$ as described in the previous section). The time cost so far is $O(n^{4g+7})$.

Recall the sets $W$, $X$, and $Y$ from the previous section. Having stored the set
 $W$ of all walks $\overrightarrow{w(k, \vec{v})}$, we compute and store the set 
\[ X = \{\overrightarrow{w_1} + \cdots + \overrightarrow{w_r}: \overrightarrow{w_i} \in W \: \forall i \:\text{ and }\: r \leq g+1 \}.
\] 
Adding elements of $W$ together requires $O(n)$ time; hence computing and storing $X$ requires $O(n|X|) = O(n|W|^{g+1}) = O(n \cdot n^{2(g+1)^2})$ time. We compute $\Theta(\sigma)$ for every $\sigma \in X$ and store the set $Y = \{ \sigma \in X: \Theta(\sigma)=\vec{0} \}$, which takes $O(n|X|)$ time. Finally we find an element $\sigma^*$ of $Y$ that minimizes $\hat{d}^f_{D(G)}$, which takes $O(n|Y|)$ time, and from Corollary~\ref{Last} we know $\sigma^*$ minimizes $\hat{d}^f_{D(G)}$ over all elements of $\mathcal{B}(D(G)$. Thus $d^f(G)$ is given by $\hat{d}^f_{D(G)}(\sigma^*)$, and the total time taken to find $\sigma^*$ is dominated by $\max(O(n^{2(g+1)^2 + 1}), O(n^{4g+7})) = O(n^{2g^2 + 4g +7})$.

In order to find an $f$-sparest cut of $G$, we compute $\phi^* = D^{-1}(\sigma^*)$ and decompose it as in the proof of Lemma~\ref{Avecut} to give
\[ \phi = \sum_{i=1}^{k}  \overrightarrow{ [S_i,\bar{S_i}] },
\]
where $k=O(mn)$ from our bound on $H^*$. Now one of the cuts $[S_i,\bar{S_i}]$ is an $f$-sparsest cut by the proof of Lemma~\ref{Avecut}, which can be found in $O(n^2)$ time once $\sigma^*$ has been found.

\section{Open Problems}
An obvious question that arises from this work is whether the running time of our algorithm can be improved. Specifically, it would be interesting to know if the problem of finding the edge expansion of a graph is fixed parameter tractable with respect to genus.

Our algorithm crucially relies on our graph being embedded on an orientable surface. In particular, we use the fact that a graph embedded on an orientable surface has a \emph{directed} dual; this is not the case for graphs embedded on non-orientable surfaces. It would be interesting to develop methods for finding edge expansion of graphs embedded on non-orientable surfaces.
\bibliographystyle{plain}
\bibliography{Vul}

\begin{thebibliography}{10}

\bibitem{AMS}
Christoph Ambuhl, Monaldo Mastrolilli, and Ola Svensson.
\newblock Inapproximability results for sparsest cut, optimal linear
  arrangement, and precedence constrained scheduling.
\newblock In {\em FOCS '07: Proceedings of the 48th Annual IEEE Symposium on
  Foundations of Computer Science}, pages 329--337, Washington, DC, USA, 2007.
  IEEE Computer Society.

\bibitem{ARV}
Sanjeev Arora, Satish Rao, and Umesh Vazirani.
\newblock Expander flows, geometric embeddings and graph partitioning.
\newblock {\em J. ACM}, 56(2):Art. 5, 37, 2009.

\bibitem{Bon}
Paul~S. Bonsma.
\newblock Linear time algorithms for finding sparsest cuts in various graph
  classes.
\newblock In {\em 6th {C}zech-{S}lovak {I}nternational {S}ymposium on
  {C}ombinatorics, {G}raph {T}heory, {A}lgorithms and {A}pplications},
  volume~28 of {\em Electron. Notes Discrete Math.}, pages 265--272. Elsevier,
  Amsterdam, 2007.

\bibitem{CEN}
Erin~W. Chambers, Jeff Erickson, and Amir Nayyeri.
\newblock Minimum cuts and shortest homologous cycles.
\newblock In {\em SCG '09: Proceedings of the 25th annual symposium on
  Computational geometry}, pages 377--385, New York, NY, USA, 2009. ACM.

\bibitem{Eric}
Jeff Erickson and Kim Whittlesey.
\newblock Greedy optimal homotopy and homology generators.
\newblock In {\em Proceedings of the {S}ixteenth {A}nnual {ACM}-{SIAM}
  {S}ymposium on {D}iscrete {A}lgorithms}, pages 1038--1046 (electronic), New
  York, 2005. ACM.

\bibitem{Gar}
Michael~R. Garey and David~S. Johnson.
\newblock {\em Computers and intractability}.
\newblock W. H. Freeman and Co., San Francisco, Calif., 1979.
\newblock A guide to the theory of NP-completeness, A Series of Books in the
  Mathematical Sciences.

\bibitem{Gib}
P.~J. Giblin.
\newblock {\em Graphs, surfaces and homology}.
\newblock Chapman \& Hall, London, second edition, 1981.
\newblock An introduction to algebraic topology, Chapman and Hall Mathematics
  Series.

\bibitem{GT}
Jonathan~L. Gross and Thomas~W. Tucker.
\newblock {\em Topological graph theory}.
\newblock Dover Publications Inc., Mineola, NY, 2001.
\newblock Reprint of the 1987 original [Wiley, New York; MR0898434 (88h:05034)]
  with a new preface and supplementary bibliography.

\bibitem{Rao}
Tom Leighton and Satish Rao.
\newblock Multicommodity max-flow min-cut theorems and their use in designing
  approximation algorithms.
\newblock {\em J. ACM}, 46(6):787--832, 1999.

\bibitem{Mat}
David~W. Matula and Farhad Shahrokhi.
\newblock Sparsest cuts and bottlenecks in graphs.
\newblock {\em Discrete Appl. Math.}, 27(1-2):113--123, 1990.
\newblock Computational algorithms, operations research and computer science
  (Burnaby, BC, 1987).

\bibitem{Moh}
Bojan Mohar.
\newblock A linear time algorithm for embedding graphs in an arbitrary surface.
\newblock {\em SIAM J. Discrete Math.}, 12(1):6--26 (electronic), 1999.

\bibitem{Tho}
Bojan Mohar and Carsten Thomassen.
\newblock {\em Graphs on surfaces}.
\newblock Johns Hopkins Studies in the Mathematical Sciences. Johns Hopkins
  University Press, Baltimore, MD, 2001.

\bibitem{Park}
James~K. Park and Cynthia~A. Phillips.
\newblock Finding minimum-quotient cuts in planar graphs.
\newblock In {\em STOC}, pages 766--775, 1993.

\bibitem{Rao2}
Satish~B. Rao.
\newblock Faster algorithms for finding small edge cuts in planar graphs.
\newblock In {\em STOC '92: Proceedings of the twenty-fourth annual ACM
  symposium on Theory of computing}, pages 229--240, New York, NY, USA, 1992.
  ACM.

\end{thebibliography}
\nocite{*}

\end{document}